\documentclass[11pt, a4paper]{article}
\usepackage{latexsym,amsmath,amssymb,amsbsy,amsthm,graphicx,euscript}
\usepackage{hyperref}
\usepackage{physics}

\renewcommand{\thefootnote}{\fnsymbol{footnote}}

\newtheorem{lemma}{Lemma}
\newtheorem{sep}{Separability criterion}
\newtheorem{example}{Example}
\newtheorem*{state}{Statement}
\newtheorem*{unth}{Theorem}
\newtheoremstyle{dotless}{}{}{\itshape}{}{\bfseries}{}{ }{}
\theoremstyle{dotless}
\newtheorem*{dotless}{}

\newcommand{\hil}{\mathcal H}
\newcommand{\scr}{\scriptscriptstyle}
\newcommand{\lsc}[1]{_{\scr #1}}
\newcommand{\vm}[1]{V_{\mu #1}}
\newcommand{\vc}[1]{\mathrm{vec}(#1)}
\newcommand{\m}{\mathfrak M}
\newcommand{\hi}[1]{^{(#1)}}

\makeatletter
\@addtoreset{footnote}{section}
\makeatother

\begin{document}
{\Large\bf Channel-state duality and the separability problem.}

\medskip
\begin{quotation}\noindent
  K.~V.~ Antipin,$^{(1)}$\footnote{kv.antipin@physics.msu.ru}
  \smallskip
  \label{b}\\
  $^{(1)}$\small{\em  Physics Department, Lomonosov Moscow State University, Moscow 119991, Russia.}

  \bigskip\medskip
  \noindent Separability of quantum states is analyzed with the use of the Choi-Jamiolkowski isomorphism. Spectral separability criteria are derived. The presented approach is illustrated with various examples, among which a separable decomposition of $2\otimes 2$ isotropic states is obtained.
\end{quotation}

\section{Introduction}
\renewcommand{\thefootnote}{\arabic{footnote}}

Entanglement as a resource~\cite{Preskill, Nielsen}  is a central notion in  quantum information theory. The important question is to tell whether a given quantum composite system state is entangled or separable. A bipartite mixed state is separable if it can be expressed as a convex combination of product states:
\begin{equation}
	\rho\lsc{AB} = \sum_i\, p_i\, \rho^i\lsc A\otimes\rho^i\lsc B,
\end{equation}
where $\rho^i\lsc A,\, \rho^i\lsc B$ -- local density matrices of subsystems $A$ and $B$ respectively; $p_i$ -- ensemble probabilities.

One of the first remarkable results in this direction was the positive partial transposition~(PPT) criterion~\cite{Peres} as necessary condition for separability of  bipartite mixed states. This simple but extremely useful observation by Asher Peres has generated further considerable research. It was proved that PPT condition is necessary and sufficient  for separability of $2\otimes2$ and $2\otimes3$ states~\cite{Horod1}. Over time several other necessary or/and sufficient criteria were developed~\cite{Horod1, Terhal_Bell, Wu, Horod2, Realign, CCN}, among which entanglement witnesses~\cite{Terhal_Bell, Horod2} and the CCNR criterion~\cite{Realign, CCN} proved to be important tools in detecting entanglement. The Bloch representation was first introduced into the separability problem by de Vicente~\cite{Vicente}, the idea, extended in a recent  paper~\cite{Nature_sep} where a necessary and sufficient separability criterion was obtained in terms of inequalities for singular values of the correlation matrix.

In this paper we present an approach to the separability problem inspired by the well-known correspondence between completely positive~(CP) maps and states, the Choi-Jamiolkowski isomorphism~\cite{Jam, Choi}. We make use of the fact that a state is separable if and only if the corresponding CP~map can have operator-sum representation with unit rank operators. Different operator-sum representations of the same CP map are connected by linear transformations of the specific type. Thus, the state is separable if and only if the operators\footnote{If a CP map is trace-preserving, they are called Kraus operators.} of the corresponding CP~map can all be transformed to those of unit rank. We analyze the conditions under which such transformations exist and show that this approach can be a powerful tool in investigating the separability problem.

The paper is organized as follows: in Section~\ref{sec2}\:\:\:, based on the properties of the Choi-Jamiolkowski isomorphism, we derive a necessary and sufficient condition for separability of bipartite mixed states; in Section~\ref{sec3}\:\:\: we consider the applications of this result: in Section~\ref{sec3}\:\:\:.\ref{sec3:sub1}\:\:\: spectral separability criteria are obtained,
in Section~\ref{sec3}\:\:\:.\ref{sec3:sub2}\:\:\: we illustrate our method with several examples including a separable decomposition of  $2\otimes 2$ isotropic states; in Section~\ref{sec4}\:\:\: the presented approach is reformulated in terms of the factorization of the density operator; finally, Section~\ref{sec5}\:\:\: contains conclusions and discussion of the open questions.

\section{\label{sec2}Channel-state duality and connections with separability}

In this section we derive a necessary and sufficient condition for separability of a bipartite mixed state.

We begin with recalling the properties of the Choi-Jamiolkowski isomorphism between states and CP~maps. Let $\rho_{\scr AB}$ be a density operator acting on Hilbert space $\hil_{AB} = \hil_{A}\otimes\hil_{B} $ with dimensions $dim(\hil_{A}) = m$, $dim(\hil_{B}) = n$, $m\le n$. Let $\sigma$ be a density operator on $\hil_{A}$. A~CP~map $\Lambda_{\rho}\colon \mathcal L(\hil_A)\to\mathcal L(\hil_B)$, corresponding to $\rho_{\scr AB}$, can be defined by the action on an arbitrary state $\sigma$ as follows~\cite{RudCJ}:
\begin{equation}
  \label{CJ}
  \Lambda_{\rho}[\sigma] = m\, \mathrm{Tr_{\scriptscriptstyle A}}\{\sigma^T\!\otimes\!I_{\scriptscriptstyle B}\; \rho_{\scr AB}\},
\end{equation}
where $\sigma^T$ -- transposed operator $\sigma$, $I_{\scriptscriptstyle B}$ -- identity operator acting on $\hil_B$.

The initial state $\rho$ is recovered by the action of the Choi operator~\cite{Wilde}:
\begin{equation}
  \label{ChoiOp}
  (I_{\scr A}\otimes \Lambda_{\rho})\ket{\Gamma}\!\bra{\Gamma}_{\scr AA} = \rho_{\scr AB},
\end{equation}
where $\ket{\Gamma}_{\scr AA}$ - the maximally entangled vector on $\hil\lsc A\otimes\hil\lsc A$:
\begin{equation}
  \label{MaxEntVec}
  \ket{\Gamma}_{\scr AA} = \frac 1{\sqrt m} \sum_{\scr i=0}^{\scr m-1} \ket{i}_{\scr A}\otimes\ket{i}_{\scr A}.
\end{equation}
Suppose that $\rho$ is realized by a specific ensemble of bipartite pure states $\ket{\Psi_a}$ with probabilities $p_a$:
\begin{equation}
  \label{StDec}
  \rho_{\scr AB} = \sum_a p_a \ket{\Psi_a}\!\bra{\Psi_a}_{\scr AB},
\end{equation}
then, from Eq.~(\ref{CJ}) and Eq.~(\ref{StDec}) it is straightforward to see that
\begin{equation}
  \label{OperEl}
  \Lambda_{\rho}[\:\ket{\varphi}\!\bra{\varphi}_{\scr A}] = m\, \sum_a p_a \,{}_{\scr A}\!\!\bra{\varphi^*}\ket{\Psi_a}\!\bra{\Psi_a}\ket{\varphi^*}\!\!_{\scr A},
\end{equation}
where we extract the action on $\dyad{\varphi_{\scr A}}$ using the dual vector ${}_{\scr A}\!\!\bra{\varphi^*}$ in accordance with the \emph{relative-state method}~\cite{PresChap3CJ}. Now, given a vector $\ket{\Psi_a}_{\scr AB}$, an operator $M_a$ mapping $\hil_A$ to $\hil_B$ can be defined by
\begin{equation}
  \label{OpElAct}
  M_a \ket{\varphi}_{\scr A} = \sqrt{m p_a}\: {}_{\scr A}\!\!\bra{\varphi^*}\ket{\Psi_a}_{\scr AB}.
\end{equation}
Eq.~(\ref{OperEl}) gives an operator-sum representation of the CP map $\Lambda_{\rho}$ acting on a pure state projector $\ket{\varphi}\!\bra{\varphi}_{\scr A}$ and hence, by linearity, on any density operator $\sigma$:
\begin{equation}
  \label{OpRep}
  \Lambda_{\rho}[\sigma] = \sum_a M_a\sigma M_a^{\dagger}.
\end{equation}

One important property of the operators $M_a$, which we will use in the present paper, is the connection with the reduced density matrices $\rho_{\scr A} = \mathrm Tr_{\scr B}\{ \rho_{\scr AB}\}$ and $\rho_{\scr B} = \mathrm Tr_{\scr A}\{ \rho_{\scr AB}\}$. Suppose that the pure states $\ket{\Psi_a}$ from the ensemble admit the following decomposition in the given orthonormal bases $\ket{i}_{\scr A}$, $\ket{j}_{\scr B}$  of $\hil_A$ and  $\hil_{\scr B}$:
\begin{equation}
  \label{PureDec}
  \ket{\Psi_a}_{\scr AB} = \sum_{i, j} c^{(a)}_{ij}\ket{i}_{\scr A}\otimes\ket{j}_{\scr B}, \:  0 \le i \le m-1,\, 0\le j \le n-1.
\end{equation}
Consider the matrix element $\mel{i}{M_a^{\dagger}M_a}{j}$: with the use of Eqs.~(\ref{OpElAct}), (\ref{PureDec}) it can be evaluated as follows:
\begin{equation}
  \label{MxOpEl}
  \begin{split}
   {}_{\scr A}\!\mel{i}{M_a^{\dagger}M_a}{j}_{\scr A} = m p_a\: {}_{\scr AB}\!\!\bra{\Psi_a}\!(\dyad{i}{j})_{\scr A}\!\ket{\Psi_a}\!\!_{\scr AB} = \\
   = \sum_{\scr q,r,s,t} m p_a\: c^{(a)*}_{qr} \ip{q}{i}_{\scr A}\bra{r}_{\scr B} \; c^{(a)}_{st} \ip{j}{s}_{\scr A}\ket{t}_{\scr B} = \\
   = \sum_t m p_a \, c^{(a)*}_{it} c^{(a)}_{jt} = m p_a (c^{(a)}c^{(a)\dagger})_{ji} = \\
   = m p_a (\rho^{(a)}_{\scr A})_{ji},
 \end{split}
\end{equation}
where $\rho^{(a)}_{\scr A} = c^{(a)}c^{(a)\dagger} = \mathrm{Tr}_{\scr B}\{\;\dyad{\Psi_a}_{\scr AB}\}$ -- reduction of $\dyad{\Psi_a}_{\scr AB}$ on subsystem $A$ such that $\rho_{\scr A} = \sum_a p_a \rho_{\scr A}^{(a)}$. We see that
\begin{equation}
  \label{OpDenCon}
  M_a^{\dagger}M_a = m p_a (\rho^{(a)}_{\scr A})^T,
\end{equation}
and the following property holds: \emph{a CP map is trace-preserving $(\sum_a M_a^{\dagger}M_a  = I)$ iff the reduced density operator $\rho\lsc{A}$ of the corresponding state is maximally mixed: $\rho\lsc A = \tfrac1m\: I\lsc A$}.

In a similar way, it can be obtained that
\begin{equation}
  \label{OpDenConB}
  M_aM^{\dagger}_a = m p_a \, \rho^{(a)}\lsc B.
\end{equation}
In addition, substituting $\tfrac1m\: I\lsc A $ for $\ket{\varphi}\!\bra{\varphi}_{\scr A}$ in Eq.~(\ref{OperEl}) gives:
\begin{equation}
  \label{Unital}
  \Lambda_{\rho}[\tfrac1m \,I\lsc A] = \mathrm{Tr}\lsc A \{\rho\lsc{AB}\} = \rho\lsc B,
\end{equation}
hence \emph{a CP map is unital iff the reduced density operator $\rho\lsc{B}$ of the corresponding state is maximally mixed: $\rho\lsc B = \tfrac1n\: I\lsc B$}.

The second important property is the transformations  of $M_a$~\cite{PresChap3CJ}. According to the Hughston-Jozsa-Wootters theorem~\cite{HJW}, two different ensemble realizations of the same density operator
\begin{equation}
  \label{TwoEns}
  \sum_a p_a \dyad{\Psi_a}\lsc{AB} = \sum_{\mu} q_{\mu} \dyad{\Phi_{\mu}}\lsc{AB}
\end{equation}
are related by
\begin{equation}
  \label{HJWTr}
  \sqrt{q_{\mu}} \ket{\Phi_{\mu}} = \sum_a \sqrt{p_a}\, V_{\mu a} \ket{\Psi_a},
\end{equation}
where $V_{\mu a}$ --- a matrix with orthonormal columns\footnote{It will be a unitary matrix if the numbers of terms in both ensembles are equal.}. Therefore, from Eqs.~(\ref{OpElAct}),~(\ref{HJWTr}) it follows that operators $N_{\mu}$ and $M_a$, corresponding to $\ket{\Phi_{\mu}}$ and $\ket{\Psi_a}$, are related by
\begin{equation}
  \label{OpTransf}
  N_{\mu} = \sum_a V_{\mu a}M_a.
\end{equation}
Eq.~(\ref{OpTransf}) plays the main role in the development of our paper. Now, to analyze the separability of quantum states, we need one more statement that can be deduced from Theorem~4 of Ref.~\cite{EntBre}~(a similar theorem can also be found in Ref.~\cite{WildEntBre}):
\begin{unth}  \label{thm:gen} Let $\Phi$ be a completely positive map.  The following are equivalent
\begin{itemize}
\item[A)] $\Phi$ has the Holevo form:
	\[
		\Phi(\rho) = \sum_k R_k \, \tr  F_k \rho
	,\] 
	where $\{R_k\}$ are density matrices and $\{F_k\}$ are positive semi-definite.
\item[B)] $\Phi$ is entanglement-breaking.
\item[C)] $(I \otimes \Phi)( \dyad{ \beta}  )$ is separable
	for $\ket{\beta}= d^{-1/2} \sum_j \ket{j} \otimes \ket{j}$,
a maximally entangled  state.
\item[D)] $\Phi$ can be written in operator sum
form using only rank one operators.
\item[E)] $\Upsilon  \circ  \Phi$ is completely positive for
all positivity preserving maps $\Upsilon$.
\item[F)] $\Phi \circ \Upsilon $ is completely positive for
all positivity preserving maps $\Upsilon$.
\end{itemize}
\end{unth}
\medskip

Based on this result, we can prove the following
\begin{lemma}\label{main_lem}
	A bipartite mixed state $\rho$ is separable if and only if the operators $\{M_a\}$ from the operator-sum representation in Eq.~(\ref{OpElAct}) of the corresponding completely positive map in Eq.~(\ref{CJ}) can all be transformed by means of Eq.~(\ref{OpTransf}) to rank one operators $\{N_{\mu}\}$.
\end{lemma}
\begin{proof}
  If $\rho$ is separable, then, according to Eq.~(\ref{ChoiOp}), the corresponding map $\Lambda_{\rho}$ transforms the maximally entangled state to a separable state, and, by the proposition~$(C)\Rightarrow(D)$ of Theorem~4 from Ref.~\cite{EntBre}, $\Lambda_{\rho}$ can be written in operator sum using only operators of rank one, i.~e., Eq.~(\ref{OpTransf}) holds with  $\{N_{\mu}\}$ of rank one. The converse statement is also true due to the equivalence of the clauses~(C) and (D) of the above mentioned Theorem.
\end{proof}

As an indirect application of Lemma~\ref{main_lem}, we consider a couple of separability criteria based on spectral properties.

\section{Applications of Lemma~\ref{main_lem}}\label{sec3}

\subsection{\label{sec3:sub1}Spectral separability criteria and some inequalities}

For now, let Eq.~(\ref{StDec}) express a spectral decomposition of the density operator $\rho$. The matrices $c^{(a)}$ of Eq.~(\ref{PureDec}), which correspond to the eigenvectors $\ket{\Psi_a}$, are orthonormal with respect to the Hilbert-Schmidt inner product:
\begin{equation}
  \label{MatHilProd}
  \mathrm{Tr} \{c^{(a) \dagger}c^{(b)}\} = \delta_{ab}.
\end{equation}
From Eqs.~(\ref{OpElAct}), (\ref{PureDec}), (\ref{MatHilProd}) it follows, then, that operators $\{M_a\}$ are mutually orthogonal:
\begin{equation}
  \label{OpOrt}
  \mathrm{Tr} \{M_a^{\dagger} M_b\} = m p_a \delta_{ab}.
\end{equation}
Now, consider Eq.~(\ref{OpTransf}). We can take advantage of the spectral decomposition and express the coefficients $V_{\mu a}$ using Eq.~(\ref{OpOrt}):
\begin{equation}
  \label{UnCoef}
  V_{\mu a} = \frac{\mathrm{Tr}\{M_a^{\dagger}N_{\mu}\}}{m p_a}.
\end{equation}
 Being the entries of a matrix with orthonormal columns, $V_{\mu a}$ satisfy the normalization condition:
\begin{equation}
  \label{ColAbs}
\sum_{\mu}  \lvert V_{\mu a}\rvert ^2 = 1.
\end{equation}
At the same time, we can give an upper bound on $\lvert\mathrm{Tr}\{M_a^{\dagger}N_{\mu}\}\rvert$ using the Cauchy-Schwarz inequality or an even sharper bound using the following property~\cite{MA,Bhatia}:
\begin{equation}
  \label{TrSing}
  \lvert\mathrm{Tr}\{A^{\dagger}B\}\rvert \leqslant \sum_{i=1}^{q} \sigma_i (A) \sigma_i (B),
\end{equation}
where $A, B$ -- complex $m\times n$ matrices, $q = \mathrm{min}\{m,n\}$, $\sigma_i(A), \,\sigma_i(B)$ -- singular values of A and B arranged in non-increasing order: $\sigma_1(A)\geqslant\sigma_2(A)\geqslant\ldots\geqslant\sigma_q(A)$.

Let $\lambda^{(a)}_1, \ldots, \lambda^{(a)}_m$ denote the eigenvalues of the density operator $\rho^{(a)}\lsc A$ on subsystem A, taken in non-increasing order. From Eq.~(\ref{OpDenCon}) we can see the connection between the singular values $\sigma_i(M_a)$ of the operators $M_a$~(thought of as matrices) and the eigenvalues $\lambda^{(a)}_i$:
\begin{equation}
  \label{ValConn}
  \sigma_i(M_a) = \sqrt{m p_a \lambda^{(a)}_i}.
\end{equation}
Eqs.~(\ref{UnCoef}), (\ref{TrSing}), (\ref{ValConn}) give an upper bound on $\lvert V_{\mu a}\rvert$:
\begin{equation}
  \label{BoundV}
  \lvert V_{\mu a}\rvert\leqslant \sum_{i = 1}^{m}\sqrt{\frac{q_{\mu}\tilde\lambda^{(\mu)}_i \lambda^{(a)}_i}{p_a}},
\end{equation}
where $q_{\mu}$ -- probabilities of the second ensemble realization of the density operator $\rho$, as in Eq.~(\ref{TwoEns}); $\tilde{\lambda}^{(\mu)}_i$ -- eigenvalues of $\tilde \rho^{(\mu)}\lsc A = \mathrm{Tr}\lsc B \{\;\dyad{\Phi_{\mu}}\lsc{AB}\}$.

If $\rho$ is separable, then, according to Lemma~\ref{main_lem}, there exist coefficients $V_{\mu a}$ such that $\{N_{\mu}\}$ are rank one operators. Correspondingly, $ \tilde \rho^{(\mu)}\lsc A$ are also of rank one, and so each of them has only one non-vanishing eigenvalue:
\begin{equation}
  \label{EigenVan}
\tilde\lambda^{(\mu)}_1 = 1, \, \tilde\lambda^{(\mu)}_2 = \ldots = \tilde\lambda^{(\mu)}_m = 0.
\end{equation}
Combining this fact with Eqs.~(\ref{ColAbs}), (\ref{BoundV}), we obtain:
\begin{equation}
  \label{SpecCrit}
  1 = \sum_{\mu}  \lvert V_{\mu a}\rvert ^2 \leqslant \sum_{\mu} \frac{q_{\mu} \lambda^{(a)}_1}{p_a} = \frac{\lambda^{(a)}_1}{p_a}\sum_{\mu} q_{\mu} = \frac{\lambda^{(a)}_1}{p_a}.
\end{equation}
As a result, we have the following
\begin{sep}\label{sep1}
If a density operator $\rho\lsc{AB}$ with a given spectral decomposition \[
\rho\lsc{AB} = \sum_a p_a \dyad{\Psi_a}\lsc{AB}
\]  is separable, then the largest eigenvalue of each operator  $\rho^{(a)}\lsc{A} = \mathrm{Tr}\lsc B \{\;\,\dyad{\Psi_a}\lsc{AB}\}$ is greater than or equal to the corresponding ensemble probability: 
\begin{equation}\label{SepCrit1}
\lambda^{(a)}_1 \geqslant\, p_a
\end{equation}
\end{sep}

\bigskip\noindent
We note that the criterion was derived independently of Ref.~\cite{SpecCrit}, where the same property was obtained by a different method within the framework of the theory of entanglement witnesses.

As an example of application of the criterion, consider isotropic states~\cite{ISO} in arbitrary dimension $d$~($m = n = d$):
\begin{equation}
  \label{Iso}
  \rho^{iso}(\alpha) = \alpha \dyad*{\Phi^+} + \frac{1-\alpha}{d^2}\,I_d\otimes I_d,
\end{equation}
where $0\leqslant\alpha\leqslant 1,$ and $\ket{\Phi^+}$ -- maximally entangled state:
\begin{equation}
  \label{MxEnt2}
  \ket{\Phi^+} = \frac1{\sqrt d}\sum_{i=0}^{d-1} \ket{i}\lsc A \otimes\ket{i}\lsc B.
\end{equation}
In order to obtain a spectral decomposition of $\rho^{iso}$, we introduce mutually orthogonal auxiliary states
\begin{equation}
  \label{Symanti}
  \begin{split}
  &\ket{\Phi^+_k} = \frac1{\sqrt d}\sum_{j = 0}^{d-1} e^{\scriptscriptstyle i 2\pi jk/d} \ket{j}\lsc A \otimes\ket{j}\lsc B, \: k = 0, \ldots, d-1, \\
  \ket{\Psi^{\pm}_{ij}} &= \frac1{\sqrt 2}(\;\ket{i}\lsc A\otimes\ket{j}\lsc B \pm \ket{j}\lsc A\otimes\ket{i}\lsc B), \, i < j, \: i,j = 0, \ldots, d-1,
  \end{split}
\end{equation}
such that $\ket{\Phi^+}\equiv\ket{\Phi^+_0}$; $d$ states $\ket{\Phi^+_k}$ and $d(d-1)/2$ states $\ket{\Psi^+_{ij}}$ belong to the symmetric subspace of $\hil\lsc{AB}$, and $d(d-1)/2$ states $\ket{\Psi^-_{ij}}$ belong to the antisymmetric subspace. Therefore, the identity operator $I_d\otimes I_d$ can be decomposed in terms of $d^2$ mutually orthogonal states from Eq.~(\ref{Symanti}):
\begin{equation}
  \label{IdDec}
  I_d\otimes I_d = \sum_{k=0}^{d-1} \dyad{\Phi^+_k} + \sum_{\substack{i<j\\i, j = 0}}^{d-1}\left[\dyad*{\Psi^+_{ij}} + \dyad*{\Psi^-_{ij}}\right],
\end{equation}
and the isotropic state density operator is expressed as follows:
\begin{equation}
  \label{IsoF}
  \begin{split}
    \rho^{iso}(F) = F \dyad*{\Phi^+_0} + &\frac{1-F}{d^2-1}\left(\sum_{k=1}^{d-1}\dyad*{\Phi^+_k}  + \sum_{\substack{i<j\\i, j = 0}}^{d-1} \left[\dyad*{\Psi^+_{ij}} + \dyad*{\Psi^-_{ij}}\right]\right),\\
    &F = \frac{\alpha (d^2 - 1) + 1}{d^2}.
  \end{split}
\end{equation}
Applying separability criterion~\ref{sep1} to the first term of the decomposition in Eq.~(\ref{IsoF}), we obtain that if $\rho^{iso}$ is separable, then
\begin{equation}
  \label{SepIso}
  \lambda_1(\dyad*{\Phi^+_0}) = \frac1{d} \geqslant F.
\end{equation}
Conversely, if $F > 1/d$, i.e., $\alpha > 1/(d+1)$, then $\rho^{iso}$ is entangled. This fact was established by application of other separability criteria~\cite{ISO}.

Although separability criterion~\ref{sep1} is able to detect all entangled isotropic states, it is not as much efficient in many other cases where it can be applicable. As an example, consider a $2\otimes2$ state
\begin{equation}
  \label{mpmix}
  \rho^{\pm} = p\dyad*{\psi^{\pm}} + (1-p)\dyad{00},
\end{equation}
where $\ket{\Psi^{\pm}} = \tfrac1{\sqrt 2}(\ket{01}\pm\ket{10})$. In this case the criterion detects entanglement only when $p > 1/2$, whereas $\rho^{\pm}$ is entangled at each $p\ne 0 $.

As a consequence of separability criterion~\ref{sep1}, the following inequality holds for a separable state:
\begin{equation}\label{sepcons}
\sum_a \lambda_1^{(a)} \geqslant 1.
\end{equation}
This inequality can also be obtained from the majorization criterion~\cite{NK}, which states that the density matrix of a separable state is majorized by both of its reductions:
\begin{equation}\label{maj}
\rho\lsc{AB} \prec\rho\lsc A,\, \rho\lsc{AB} \prec\rho\lsc B,
\end{equation}
where majorization $A\prec B$ for arbitrary Hermitian operators $A$ and  $B$ is defined in terms of their eigenvalues $\lambda_1 (A)\geqslant\ldots\geqslant\lambda_d (A)$ and $\lambda_1 (B)\geqslant\ldots\geqslant\lambda_d (B)$:
\begin{equation}\label{majn}
	\sum_{j=1}^k \lambda_j (A) \leqslant\sum_{j=1}^k \lambda_j(B),
\end{equation}
for $k = 1,\ldots, d-1$, and with the inequality holding with equality when  $k=d$. Applying Eq.~(\ref{maj}) and Eq.~(\ref{majn})~(with $k=1$) to the spectral decomposition of $\rho\lsc{AB}$, in which $\lambda_j (\rho\lsc{AB})=p_j$, we obtain:
\begin{equation}\label{KFanC}
p_1\leqslant\lambda_1 (\rho\lsc A) = \lambda_1 \left(\sum_a p_a \rho^{(a)}\lsc A\right)\leqslant\sum_a p_a\lambda_1^{(a)}.
\end{equation}
Here the last less or equal sign is due to the Ky Fan's maximum principle~\cite{BhatFan} in application to a sum of operators. Combination of Eq.~(\ref{KFanC}) with the fact that $p_a\leqslant p_1$ for all $a$ yields inequality in Eq.~(\ref{sepcons}).

A tighter lower bound for $\lambda^{(a)}_1$ can be expressed in terms of the coefficients $V_{\mu a}$ that transform the initial operators $\{M_a\}$ to the ones of unit rank. Here the inequality for a perturbation bound on all singular values~\cite{HornFrob} can be used:
\begin{equation}\label{Frob}
	\left[\sum_{i=1}^q\, [\sigma_i (A) - \sigma_i (B)]^2\right]^{\tfrac12}\,\leqslant\, \norm{A-B}_2,
\end{equation}
where $A$ and $B$ -- arbitrary $m\times n$ matrices, $q=\min\{m,\,n\}$, $\norm{A}_2 = [\mathrm{Tr}\{A^{\dagger}A\}]^{1/2}$ is the Frobenius norm. To apply the inequality, we take a unit rank operator $N_{\mu}$ as $A$ and one of the terms in the right part of Eq.~(\ref{OpTransf}), $V_{\mu k}M_k$, as $B$. In the right part of the inequality we have a Frobenius norm of the operator $N_{\mu}^{(k)}\equiv N_{\mu} - V_{\mu k}M_k$, and the orthogonal condition in Eq.~(\ref{OpOrt}) can be used to calculate $\mathrm{Tr}\{N_{\mu}^{(k)\dagger}N_{\mu}^{(k)}\}$:
\begin{equation}\label{TrOrth}
	\mathrm{Tr}\{N_{\mu}^{(k)\dagger}N_{\mu}^{(k)}\} = \sum_{\substack{a\ne k\\ b\ne k }}\, V_{\mu b}^* V_{\mu a} \mathrm{Tr}\{M_b ^{\dagger}M_a\} = m\sum_{a\ne k}\,\absolutevalue{V_{\mu a}}^2 p_a.
\end{equation}
Eq.~(\ref{Frob}) in combination with Eqs.~(\ref{ValConn}),~(\ref{EigenVan}) then gives:
\begin{equation}\label{FrobOp}
	\left[\left(\sqrt{q_{\mu}} - \absolutevalue{V_{\mu k}}\sqrt{p_k\lambda_1^{(k)}}\right)^2 + \absolutevalue{V_{\mu k}}^2\,p_k\,(\lambda_2^{(k)} + \ldots + \lambda_m^{(k)})\right]^{\tfrac12}\,\leqslant\,\left(\sum_{a\ne k}\,\absolutevalue{V_{\mu a}}^2 p_a\right)^{\tfrac12}.
\end{equation}
Taking into account the equality $\lambda_2^{(k)} + \ldots + \lambda_m^{(k)} = 1 - \lambda_1^{(k)}$ and simplifying Eq.~(\ref{FrobOp}), we obtain:
\begin{equation*}
	q_{\mu} - 2\absolutevalue{V_{\mu k}}\sqrt{q_{\mu}p_k\lambda_1 ^{(k)}} + \absolutevalue{V_{\mu k}}^2 p_k\,\leqslant\,\sum_{a\ne k}\,\absolutevalue{V_{\mu a}}^2 p_a.
\end{equation*}
Summing the last inequality over $\mu$, using Eq.~(\ref{ColAbs}) and the fact that $p_a$ and $q_{\mu}$ -- ensemble probabilities, we establish the following inequality:
\begin{equation}\label{SchIneq}
	\lambda_1^{(k)}	\,\geqslant\, p_k\,\left(\sum_{\mu}\absolutevalue{V_{\mu k}}\sqrt{q_{\mu}}\right)^{-2}.
\end{equation}
We emphasize the conditions under  which Eq.~(\ref{SchIneq}) is valid: $\lambda_1^{(k)}$ and $p_k$ are associated with the spectral decomposition of the separable density operator, $V_{\mu a}$ -- a matrix transforming the spectral decomposition of the state to the separable one, $q_{\mu}$ -- ensemble probabilities of the separable decomposition. 

Eq.~(\ref{SchIneq}) expresses a tighter than in Eq.~(\ref{SepCrit1}) lower bound for $\lambda_1^{(k)}$ since, by the Cauchy-Schwarz inequality,
\[
\left(\sum_{\mu}\absolutevalue{V_{\mu k}}\sqrt{q_{\mu}}\right)^2\,\leqslant\,\sum_{\mu}\absolutevalue{V_{\mu k}}^2\,\sum_{\mu}q_{\mu} = 1
.\] 
Therefore, separability criterion~\ref{sep1} is a consequence of Eq.~(\ref{SchIneq}).

We can also apply Eq.~(\ref{Frob}) to an arbitrary density operator $\rho\lsc{AB}$, not necessarily a separable one, to obtain inequalities relating ensemble probabilities $q_{\mu}$ and eigenvalues $\tilde\lambda_i ^{(\mu)}$ of some arbitrary ensemble decomposition of $\rho\lsc{AB}$ with probabilities $p_k$ and eigenvalues $\lambda_i ^{(k)}$ of the spectral decomposition. Making similar calculations as above with the use of Eqs.~(\ref{Frob}) and (\ref{TrOrth}), we arrive at inequality
\begin{equation}\label{FrobGen}
	\absolutevalue{V_{\mu k}}^2 p_k\,\leqslant\, \absolutevalue{V_{\mu k}}\,\sum_{i=1}^m \,\sqrt{q_{\mu} p_k \lambda_i ^{(k)} \tilde\lambda_i ^{(\mu)}}.
\end{equation}
In the right part of this inequality $\absolutevalue{V_{\mu k}}$ can be eliminated with the use of Eq.~(\ref{BoundV}), in the left part we have two options of eliminating $\absolutevalue{V_{\mu k}}$: by summing over $\mu$ using Eq.~(\ref{ColAbs}) and by summing over $k$ using the equality $q_{\mu} = \sum_{k}\,\absolutevalue{V_{\mu k}}^2 p_k$. The first choice leads to inequality
\begin{equation}\label{Pineq}
	p_k\,\leqslant\,\sum_{\mu}\, q_{\mu}\,\left(\sum_{i=1}^m\,\sqrt{\lambda_i ^{(k)}\tilde\lambda_i^{(\mu)}}\right)^2,
\end{equation}
for any $k$; the second choice yields:
\begin{equation}\label{Leq}
	1\,\leqslant\,\sum_k\,\left(\sum_{i=1}^m\, \sqrt{\lambda_i ^{(k)}\tilde\lambda_i^{(\mu)}}\right)^2,
\end{equation}
for any $\mu$.

When the density operator $\rho\lsc{AB}$ is separable and $\tilde\lambda_i ^{(\mu)}$ are associated with a separable decomposition of $\rho\lsc{AB}$, Eq.~(\ref{EigenVan}) holds, and Eqs.~(\ref{Pineq}) and (\ref{Leq}) transform to Eqs.~(\ref{SepCrit1}) and (\ref{sepcons}) respectively.

Another spectral criterion can be found by analyzing inequalities for singular values of a sum of several matrices. According to the singular value analogs of the Weyl inequalities for eigenvalues of sums of Hermitian matrices~\cite{BhatWeyl}, there is an upper bound on singular values of a sum of two arbitrary $m\times n$ matrices $A$ and $B$~\cite{BhatSing}:
\begin{equation}
  \label{Weyl}
  \sigma_{i+j-1}(A + B)\leqslant \sigma_i(A) + \sigma_j(B),
\end{equation}
where $1\leqslant i, j \leqslant \mathrm{min}\{m, n\}$ and $i+j\leqslant \mathrm{min}\{m, n\} +1$. Taking $j = 1$, making substitutions $A\to A +B$, $B\to -B$ and using the fact that $\sigma_i(-B) = \sigma_i(B)$, we obtain a lower bound on $\sigma_i(A+B)$:
\begin{equation}
  \label{WeylLow}
  \sigma_i(A + B)\geqslant\sigma_i(A) - \sigma_1(B).
\end{equation}
The generalization of this inequality on an arbitrary number $l$ of matrices reads as follows:
\begin{equation}
  \label{WeylGen}
  \sigma_i(A_1 + \ldots + A_l)\geqslant\sigma_i(A_k) - \sum_{j\ne k}\sigma_1(A_j),
\end{equation}
where $k$ -- any integer between $1$ and $l$. Applying inequality in Eq.~(\ref{WeylGen}) to the sum in the right part of Eq.~(\ref{OpTransf}) and taking into account Eq.~(\ref{ValConn}), we obtain:
\begin{equation}
  \label{WeylCor1}
  \sqrt{q_{\mu}\tilde\lambda^{(\mu)}_i}\geqslant\lvert V_{\mu k}\rvert\sqrt{p_k\lambda_i^{(k)}} - \sum_{a\ne k}\lvert V_{\mu a}\rvert\sqrt{p_a\lambda^{(a)}_1}.
\end{equation}
With the use of Eq.~(\ref{BoundV}) the last inequality can be transformed to the form:
\begin{equation}
  \label{WeylCor2}
  \sqrt{\tilde\lambda^{(\mu)}_i}\geqslant\lvert V_{\mu k}\rvert\sqrt{\lambda_i^{(k)} p_k/q_{\mu}} - \sum_{a\ne k} \sqrt{\lambda^{(a)}_1}\sum_{j=1}^m \sqrt{\lambda_j^{(a)}\tilde\lambda^{(\mu)}_j}.
\end{equation}
When the state $\rho\lsc{AB}$ is separable, by Lemma~\ref{main_lem} there exists a transformation with coefficients $V_{\mu a}$ such that Eq.~(\ref{EigenVan}) holds. In this case, taking $i=2$ and supposing that $\lambda^{(k)}_2\ne 0$ for chosen $k$~\footnote{if $\lambda^{(k)}_2 = 0$ for all $k$ in the ensemble, then the case is trivial, and $\rho\lsc{AB}$ is separable}, we obtain a bound for $\lvert V_{\mu k}\rvert$:
\begin{equation}
  \label{ModVBound}
  \lvert V_{\mu k}\rvert\leqslant\sqrt{q_{\mu}/(p_k\lambda^{(k)}_2)}\:\sum_{a\ne k}\lambda^{(a)}_1.
\end{equation}
Squaring the last equation and then summing it over $\mu$, we have:
\begin{equation}
  \label{SecRawCrit}
  \sum_{\mu}\lvert V_{\mu k}\rvert^2 = 1\; \leqslant\; \frac1{p_k\lambda^{(k)}_2} \left(\sum_{a\ne k}\lambda^{(a)}_1\right)^2.
\end{equation}
The second separability criterion can be formulated.
\begin{sep}\label{sepWeyl}
	Let $\rho_{\scr AB} = \sum_a p_a \ket{\Psi_a}\!\bra{\Psi_a}_{\scr AB}$ express a spectral decomposition of the density operator $\rho\lsc{AB}$ and $\lambda^{(a)}_i$ denote the $i$-th eigenvalue of  $\rho^{(a)}\lsc{A} = \mathrm{Tr}\lsc B \{\;\,\dyad{\Psi_a}\lsc{AB}\}$. If $\rho\lsc{AB}$ is separable, then for all $k$ such that $\lambda^{(k)}_2\ne 0$, the inequality
  \begin{equation}
    \label{WeylCrit}
    \left(\sum_{a\ne k}\lambda^{(a)}_1\right)^2\;\geqslant\;p_k\lambda^{(k)}_2
  \end{equation}
  holds.
\end{sep}
The criterion is applicable only when $\sum_{a\ne k}\lambda^{(a)}_1< 1$; otherwise, the inequality in Eq.~(\ref{WeylCrit}) is trivial and doesn't give any information.

Separability criterion~\ref{sepWeyl} is able to detect entanglement in some cases where separability criterion~\ref{sep1} can't. Let us suppose that inequality in Eq.~(\ref{SepCrit1}) is not violated, i.~e., $\lambda_1^{(k)}\geqslant p_k$ for all $k$. It is easy to see that separability criterion~\ref{sepWeyl} detects entanglement when for some $k$ the following conditions are satisfied:
\begin{subequations}
	\begin{align}
	1-p_k\geqslant\lambda_2^{(k)}&>\frac {1}{p_k}\left(\sum_{a\ne k}\lambda_1^{(a)}\right)^2\geqslant\frac {(1-p_k)^2}{p_k},\label{crita}\\
	 &p_k > 1/2.\label{critb}
	\end{align}
\end{subequations}
The first inequality in Eq.~(\ref{crita}) is due to the fact that $1 - \lambda_2^{(k)}\geqslant\lambda_1^{(k)}\geqslant p_k$; the last inequality in Eq.~(\ref{crita}) holds since
\[
\sum_{a\ne k}\lambda_1^{(a)}\geqslant\sum_{a\ne k}p_a=1-p_k
.\] 
Eq.~(\ref{critb}) is written  to guarantee that $1-p_k > (1-p_k)^2/p_k$.
\bigskip

As a simple example, we consider a $5\otimes 5$ state
\[
\rho = \frac23 \dyad{\psi_1} + \frac13 \dyad{\psi_2}
,\] 
where $\ket{\psi_1}=\sqrt{\frac23}\ket{00} + \sqrt{\frac13}\ket{11}$ and $\ket{\psi_2}=\sqrt{\frac13}(\ket{22} + \ket{33} + \ket{44})$. In this case
\[
	\lambda_2^{(1)} =\lambda_2 (\dyad{\psi_1})=\frac13 \:>\:  \frac1{p_1}\left(\sum_{a\ne 1}\lambda_1 ^{(a)}\right)^2 = \frac {\left(\lambda_1 ^{(2)}\right)^2}{p_1} = \frac{(1/3)^2}{2/3} = \frac16
,\] 
and, while none of the inequalities of separability criterion~\ref{sep1} is violated, inequality of Eq.~(\ref{WeylCrit}) doesn't hold for $\lambda_2 ^{(1)}$, so separability criterion~\ref{sepWeyl} says that $\rho$ is entangled.

\subsection{\label{sec3:sub2}Direct application to concrete states}

Lemma~\ref{main_lem} states that a bipartite density operator $\rho\lsc{AB}$ is separable if and only if the corresponding completely positive map can be represented by unit rank operators $\{N_{\mu}\}$. Therefore, given \emph{any} ensemble decomposition of $\rho\lsc{AB}$~(not necessarily a spectral one) and a corresponding set \{$M_a\}$ of initial operators representing the CP map, by virtue of Eq.~(\ref{OpTransf}) one can search for linear combinations of $M_a$ with coefficients $V_{\mu a}$ that produce unit rank operators. To do this, one can set the determinants of all second order minors of $N_{\mu}$ to null and analyze the resulting system of polynomial equations in variables $V_{\mu a}$. $\rho\lsc{AB}$ is separable if and only if there exist solutions $V_{\mu a}$ which form a matrix with orthonormal columns.

We consider several known examples in which separability was earlier analyzed by other methods.

\medskip
\begin{example}\label{exmix}
	A mixture of two $2\otimes 2$ maximally entangled density operators
\begin{equation}\label{mix}
	\rho = p \dyad*{\psi^+} + (1-p)\dyad*{\phi^-},
\end{equation}
where $\ket{\psi^+} = \frac1{\sqrt 2}(\ket{01} + \ket{10})$ and $\ket{\phi^-} =\frac1{\sqrt 2}(\ket{00} - \ket{11})$.
\end{example}

\noindent According to Eq.~(\ref{OpElAct}), the CP map corresponding to $\rho$ can be represented by two operators  $M_{\psi^+}$ and $M_{\phi^-}$ related to the states $\ket{\psi^+}$ and  $\ket{\phi^-}$.  Their matrix representations are obtained by substituting the elements of the computational basis for $\ket{\varphi}$ and the states $\ket{\psi^+}$ and  $\ket{\phi^-}$ for $\ket{\Psi_a}$ in Eq.~(\ref{OpElAct}):
\begin{equation}
	M_{\psi^+} = \sqrt p \begin{pmatrix}0 & 1\\ 1 & 0\end{pmatrix},\: M_{\phi^-} = \sqrt{1-p} \begin{pmatrix}1 & 0\\0 & -1\end{pmatrix}.
\end{equation}
Eq.~(\ref{OpTransf}) gives the expression for $N_{\mu}$ in terms of $V_{\mu a}$:
\begin{equation}\label{nexp}
	N_{\mu} = \begin{pmatrix}\sqrt{1-p}\,V_{\mu 2} & \sqrt p \,V_{\mu 1}\\ \sqrt p \,V_{\mu 1} & -\sqrt{1-p}\,V_{\mu 2}\end{pmatrix}.
\end{equation}
Operator $N_{\mu}$ is of unit rank when the determinant of its matrix vanishes:
\begin{equation}\label{detvan1}
	(1-p)V_{\mu 2}^2 = - p V_{\mu 1}^2 .
\end{equation}
Taking absolute value of both parts of Eq.~(\ref{detvan1}), one can see that it is consistent with the normalization condition of Eq.~(\ref{ColAbs}) when $1-p = p$. Therefore, \emph{the density operator in Eq.~(\ref{mix}) is entangled when $p \ne 1/2$}. 
When $p=1/2$, an example of a unitary matrix $V$ satisfying $V_{\mu 2}^2 = - V_{\mu 1}^2$ can easily be found:
\begin{equation}\label{vmat}
	V = \frac1{\sqrt 2}\begin{pmatrix}1 & i\\ 1 & -i\end{pmatrix}.
\end{equation}
	Therefore, \emph{when $p = 1/2$,  $\rho$ is separable}. The proposed method also gives a separable decomposition in this case. Substituting the entries from the first row~($\mu = 1$) of the matrix of Eq.~(\ref{vmat}) into Eq.~(\ref{nexp}), we obtain $N_1 = \frac1{2}\bigl(\begin{smallmatrix}i & 1\\ 1 & -i\end{smallmatrix}\bigr)$, and, by Eqs.~(\ref{OpDenCon}), (\ref{OpDenConB})~(with $N_1$ instead of $M_a$), the reductions
\begin{equation}
	\rho^{(1)}_{\scr A} =\rho^{(1)}_{\scr B} = \frac1{2}\begin{pmatrix}1 & i\\ -i & 1\end{pmatrix}
\end{equation}
along with the ensemble probability $q_1 = \frac1m\mathrm{Tr}\{N_1 ^{\dagger} N_1\} = 1/2$.

\noindent In a similar way, it can be obtained that
\begin{equation}
	\rho^{(2)}_{\scr A} =\rho^{(2)}_{\scr B} = \frac1{2}\begin{pmatrix}1 & -i\\ i & 1\end{pmatrix},\, q_2 = 1/2,
\end{equation}
and so the separable decomposition~(one of the many possible) of $\rho$ is:
\begin{equation}
	\rho = \frac12 \rho^{(1)}_{\scr A} \otimes\rho^{(1)}_{\scr B} + \frac12 \rho^{(2)}_{\scr A}\otimes\rho^{(2)}_{\scr B}.  
\end{equation}

\medskip
\begin{example}\label{expure}
A mixture of a maximally entangled density operator and a pure one
\begin{equation}
	  \label{mixclear}
	    \rho = p\dyad*{\psi^+} + (1-p)\dyad{00},
    \end{equation}
    where $\ket{\psi^+} = \tfrac1{\sqrt 2}(\ket{01} + \ket{10})$.
\end{example}
\noindent In this case the CP map is represented by two operators $M_{\psi^+}$ and $M_{00}$ :
\begin{equation}
	M_{\psi^+} = \sqrt p \begin{pmatrix}0 & 1\\ 1 & 0\end{pmatrix},\, M_{00} = \sqrt{2(1-p)} \begin{pmatrix}1 & 0\\ 0 & 0\end{pmatrix}.
\end{equation}
Operators $N_{\mu}$, their linear combinations, are given by
\begin{equation}
	N_{\mu} = \begin{pmatrix}\sqrt{2(1-p)}\,V_{\mu 2} & \sqrt p \,V_{\mu 1}\\ \sqrt p \,V_{\mu 1} & 0\end{pmatrix}.
\end{equation}
Again, $N_{\mu}$ will be of unit rank when the determinant of its matrix vanishes:
\begin{equation}
	\mathrm{det}\, N_{\mu} = -p V_{\mu 1}^2 = 0.
\end{equation}
The coefficients $V_{\mu 1}$ cannot be equal to null simultaneously due to Eq.~(\ref{ColAbs}), so \emph{the only possibility for $\rho$ to be separable is  $p = 0$. When $p\ne 0$,  $\rho$ is entangled}.

As Examples~\ref{exmix} and~\ref{expure} show, when the number of terms in the ensemble decomposition of a given density operator is small, the application of Lemma~\ref{main_lem} is quite easy. The analysis of the next example will take much more effort.

\medskip
\begin{example}\label{exiso}
	A qubit-qubit isotropic state $\rho^{iso}$.
\end{example}
\noindent Ensemble decomposition of $\rho^{iso}$, a $2\otimes 2$ density operator, is given by Eq.~(\ref{IsoF}) with $d = 2$ and $1/4\leqslant F\leqslant 1$. It consists of 4 terms determined by the states $\ket{\Phi_0^+},\,\ket{\Phi_1^+},\,\ket{\Psi_{01}^+},\,\ket{\Psi_{01}^-}$. Let $M_1, \,M_2, \,M_3, \,M_4$ denote the respective operators connected to these states by Eq.~(\ref{OpElAct}). Their matrix representations are:
\begin{equation}
\begin{split}
	M_1 = \sqrt F &\begin{pmatrix}1 & 0\\0 & 1\end{pmatrix},\, M_2 = \sqrt{\tfrac{1-F}3}\begin{pmatrix}1 & 0\\0 & -1\end{pmatrix},\\
	M_3 = \sqrt{\tfrac{1-F}3} &\begin{pmatrix}0 & 1\\1 & 0\end{pmatrix},\, M_4 = \sqrt{\tfrac{1-F}3}\begin{pmatrix}0 & -1\\1 & 0\end{pmatrix}.
\end{split}
\end{equation}
Operators $\{M_i\}$ give the representation of a CP map corresponding to  $\rho^{iso}$. Their transformations, $\{N_{\mu}\}$, are given by
\begin{equation}\label{isotransf}
	N_{\mu} = \begin{pmatrix}\sqrt F\, V_{\mu 1} + \sqrt{\tfrac{1-F}3}\, V_{\mu 2} & \sqrt{\tfrac{1-F}3}\, (V_{\mu 3} - V_{\mu 4})\\
		\sqrt{\tfrac{1-F}3}\, (V_{\mu 3} + V_{\mu 4}) & \sqrt F\, V_{\mu 1} - \sqrt{\tfrac{1-F}3}\, V_{\mu 2} 
	\end{pmatrix}.
\end{equation}
The rank of $N_{\mu}$ is $1$ when the determinant vanishes:
 \begin{equation}\label{detiso}
	F\, V_{\mu 1}^2 = \frac{1-F}3 \,(V_{\mu 2}^2 + V_{\mu 3}^2 - V^2_{\mu 4}).
\end{equation}
Taking absolute value of both parts of this equation gives inequality
\begin{equation}
	F\, \lvert V_{\mu 1}\rvert^2 \leqslant \frac{1-F}3 \,(\lvert V_{\mu 2}\rvert ^2 + \lvert V_{\mu 3}\rvert ^2 + \lvert V_{\mu 4}\rvert ^2),
\end{equation}
which, after summing over $\mu$ and using Eq.~(\ref{ColAbs}), takes form:
\begin{equation}
	F \leqslant \,\frac{1-F}3\, 3 = 1 - F.
\end{equation}
Consequently, \emph{when $F > 1/2$, a two-qubit isotropic state is entangled}. 

For a complete analysis, it is necessary to show that if  $1/4\leqslant F\leqslant 1/2$, then $\rho^{iso}$ is separable. Finding a matrix satisfying Eq.~(\ref{detiso}) along with the orthonormal condition for columns is not an easy task if one tries to approach the problem by straightforward solving a system of polynomial equations. We will use a method for construction of unitary matrices described in \cite{UnitConstruct}:
\begin{state}
	Let $A_1,A_2,\ldots,A_n$ be $m\times m$ unitary matrices and let
	\[ \left ( a_{ij} \right )_{i,j=1}^n,\mbox{ be a unitary matrix.}\]
	Then the following matrix is a $nm\times nm$ unitary matrix:
	\[ B=\begin{bmatrix}
		a_{11}A_1& a_{12}A_2 & \cdots &a_{1n}A_n\\
		a_{21}A_1&a_{22}A_2& \cdots &a_{2n}A_n\\
		\vdots & \vdots& \vdots& \vdots\\
		a_{n1}A_1&a_{2n}A_2 & \cdots & a_{nn}A_n
	\end{bmatrix}\]

\end{state}

\medskip
With the use of this statement we will show that for any $F\colon\, 1/4\leqslant F\leqslant 1/2$, a unitary matrix with entries $V_{\mu a}$ satisfying Eq.~(\ref{detiso}) can be constructed.

At first, we consider some particular case of Eq.~(\ref{detiso}): for example, it could be 
\begin{equation}\label{2det}
	2\, V_{\mu 1}^2 =V_{\mu 2}^2 + V_{\mu 3}^2 - V^2_{\mu 4}.
\end{equation}
This choice corresponds to $F = 2/5$. Let $U,\, W$ and  $\bigl(\begin{smallmatrix}a & b\\c & d\end{smallmatrix}\bigr)$ --- $2\times 2$ unitary matrices.  We construct our solution, the $4\times 4$ unitary matrix $V$, in accordance with the above statement: 
\begin{equation}
	V = \begin{pmatrix}a\,U & b\, W\\ c\, U & d\, W\end{pmatrix}.
\end{equation}
	The key is to keep all the coefficients as simple as possible. We can choose $a,\, b,\, c,\, d$ to have the same absolute values:  $\bigl(\begin{smallmatrix}a & b\\c & d\end{smallmatrix}\bigr) = \tfrac1{\sqrt 2}\bigl(\begin{smallmatrix}1 & i\\1 & -i\end{smallmatrix}\bigr)$.
Substituting the ansatz for $V$ into Eq.~(\ref{2det}), we obtain the following equations for the entries of $U$ and  $W$:
 \begin{equation}\label{sysiso}
	\begin{split}
	       	&(2U_{11}^2 - U_{12}^2) - (W_{12}^2 - W_{11}^2) = 0,\\
		&(2U_{21}^2 - U_{22}^2) - (W_{22}^2 - W_{21}^2) = 0.
	\end{split}
\end{equation}
If we choose $U$ in the simplest form:  $U = \tfrac1{\sqrt 2}\bigl(\begin{smallmatrix}1 & -1\\1 & 1\end{smallmatrix}\bigr)$, then the matrix $W$ satisfying Eq.~(\ref{sysiso}) can also be easily guessed:
$W = \frac12\bigl(\begin{smallmatrix}i\sqrt 3 & -i\\1 & \sqrt 3\end{smallmatrix}\bigr)$. The solution for $F = 2/5$ then will be:
\begin{equation}
	V = \frac12 \begin{pmatrix}
		1 & -1 & -\sqrt{\frac32} & \frac1{\sqrt 2}\\
		1 &  1 & \frac{i}{\sqrt 2} & \sqrt{\frac32}i \\
		1 & -1 & \sqrt{\frac32} &  -\frac1{\sqrt 2}\\
		1 &  1 & -\frac{i}{\sqrt 2} & -\sqrt{\frac32}i 
	\end{pmatrix}.
\end{equation}
We can see a specific pattern here in V:
\begin{equation}\label{isopat}
	V = \frac12 \begin{pmatrix}
		1 & -1 & -\alpha & \beta\\
		1 &  1 & i\beta & i\alpha\\
		1 & -1 & \alpha &  -\beta\\
		1 &  1 & -i\beta & -i\alpha 
	\end{pmatrix},
	\end{equation}
	where $\alpha,\, \beta$ -- some positive numbers satisfying the orthonormal condition for columns~(and rows) of $V$:  $\alpha^2 + \beta^2 = 2$. This pattern works for the general solution, for any $F,\, 1/4\leqslant F\leqslant 1/2$: substituting $V$ from Eq.~(\ref{isopat}) into Eq.~(\ref{detiso}), we can easily solve the resulting equations in variables $\alpha$ and  $\beta$. The solution is as follows:
\begin{equation}\label{geniso}
	V =\frac12 \begin{pmatrix}
		 1 & -1 & -  \sqrt{\frac{2 F+1}{2-2 F}} &   \sqrt{\frac{3-6 F}{2-2 F}} \\
		  1 & 1 &   i \sqrt{\frac{3-6 F}{2-2 F}} &   i \sqrt{\frac{2 F+1}{2-2 F}} \\
		   1 & -1 &   \sqrt{\frac{2 F+1}{2-2 F}} & -  \sqrt{\frac{3-6 F}{2-2 F}} \\
		    1 & 1 & -  i \sqrt{\frac{3-6 F}{2-2 F}} & -  i \sqrt{\frac{2 F+1}{2-2 F}} 
	    \end{pmatrix}.
\end{equation}
A unitary matrix satisfying Eq.~(\ref{detiso}) is found, and  we obtain that \emph{if $1/4\leqslant F\leqslant 1/2$, then $\rho^{iso}$ is separable}.

With the help of Eq.~(\ref{geniso}) a separable decomposition of $\rho^{iso}$ can be obtained. First, substituting $V_{\mu a}$ from Eq.~(\ref{geniso}) into Eq.~(\ref{isotransf}), we obtain the expressions for 4 unit rank operators $\{N_{\mu}\}$. Next, direct calculations with the use of Eqs.~(\ref{OpDenCon}),~(\ref{OpDenConB}), and the expression for the ensemble probabilities $q_{\mu}~=~\frac1m\mathrm{Tr}\{N_{\mu} ^{\dagger} N_{\mu}\}$ show that when $1/4\leqslant F\leqslant 1/2$,  $\rho^{iso}$ can be decomposed as:
\begin{equation}
	\rho^{iso} = \frac14 (\rho^{(1)}\lsc A\otimes\rho^{(1)}\lsc B + \rho^{(2)}\lsc A\otimes\rho^{(2)}\lsc B + \rho^{(3)}\lsc A\otimes\rho^{(3)}\lsc B + \rho^{(4)}\lsc A\otimes\rho^{(4)}\lsc B),
\end{equation}
where $\rho^{(i)}\lsc A = \dyad*{\psi^{(i)}\lsc A}$, $\rho^{(i)}\lsc B = \dyad*{\psi^{(i)}\lsc B}$ -- projectors on pure states
\begin{equation}
	\begin{array}{ll}
		\ket*{\psi^{(1)}\lsc A} = a\ket*{0}\lsc A - b\ket*{1}\lsc A,& \ket*{\psi^{(1)}\lsc B} = c\ket*{0}\lsc B - d\ket*{1}\lsc B,\\
		\ket*{\psi^{(2)}\lsc A} = b\ket*{0}\lsc A - ia\ket*{1}\lsc A,& \ket*{\psi^{(2)}\lsc B} = d\ket*{0}\lsc B + ic\ket*{1}\lsc B,\\
		\ket*{\psi^{(3)}\lsc A} = a\ket*{0}\lsc A + b\ket*{1}\lsc A,& \ket*{\psi^{(3)}\lsc B} = c\ket*{0}\lsc B + d\ket*{1}\lsc B,\\
		\ket*{\psi^{(4)}\lsc A} = b\ket*{0}\lsc A + ia\ket*{1}\lsc A,& \ket*{\psi^{(4)}\lsc B} = d\ket*{0}\lsc B - ic\ket*{1}\lsc B,
	\end{array}
\end{equation}
with
\begin{equation}
	\begin{array}{c}

	a = \sqrt{\left(3 -\sqrt{3-12 F^2}-2 \sqrt{3} \sqrt{(1-F) F}\right)/6},\\
	b = \sqrt{\left(\sqrt{3-12 F^2}+2 \sqrt{3} \sqrt{(1-F) F}+3\right)/6},\\
	c = \sqrt{\left(\sqrt{3-12 F^2}-2 \sqrt{3} \sqrt{(1-F) F}+3\right)/6},\\
	d = \sqrt{\left(3 -\sqrt{3-12 F^2}+2 \sqrt{3} \sqrt{(1-F) F}\right)/6}.

	\end{array}
\end{equation}

\bigskip
Being $2\otimes 2$ states, the density operators from examples~\ref{exmix},~\ref{expure},~\ref{exiso} could be analyzed with the use of the Peres-Horodecki criterion~(when the separable decomposition is not needed). In the next example we consider a $3\otimes 3$ state having a positive partial transpose.
\begin{example}
	Detecting entanglement of a $3\otimes 3$ PPT state.
\end{example}
Let $\rho$ be a  $3\otimes 3$ density operator constructed with the use of the unextendible product bases~(UPB) method~\cite{UPB,ExPPT}:
\begin{equation}
	\rho=\frac{1}{4}\left(I_3\otimes I_3-\sum_{i=1}^4|{\psi}_i\rangle\langle{\psi}_i|-|{S}\rangle\langle {S}|\right),
\end{equation}
where the states $\ket{\psi_i}$ and $\ket{S}$ are given by
\begin{equation}\label{UPB-3}
	\begin{array}{c}
		|\psi_{1}\rangle=\frac1{\sqrt 2}\,|0\rangle|0 - 1\rangle,~~
		|\psi_{2}\rangle=\frac1{\sqrt 2}\,|2\rangle|1 - 2\rangle,\\[1ex]
                                                
		|\psi_{3}\rangle=\frac1{\sqrt 2}\,|0 - 1\rangle|2\rangle,~~
		|\psi_{4}\rangle=\frac1{\sqrt 2}\,|1 - 2\rangle|0\rangle,\\[1ex]

		|S\rangle=\frac13\,|0+1+2\rangle|0+1+2\rangle.\\
	\end{array}
\end{equation}
The density operator $\rho$ has a positive partial transpose, but in a $3\otimes 3$ case the Peres-Horodecki criterion is not sufficient for the separability of the state.

One possible ensemble decomposition of $\rho$ can be given as $\rho=1/4\sum_{i=1}^4|\phi_i\rangle\langle\phi_i|$ \cite{ExPPT}, where
\begin{equation}\label{UPB-3-ent}
	\begin{array}{l}
		|\phi_{1}\rangle =\frac12\,(|\psi_{5}\rangle+|\psi_{6}\rangle-|\psi_{7}\rangle-|\psi_{8}\rangle),\\[1ex]
		|\phi_{2}\rangle =\frac12\,(|\psi_{5}\rangle-|\psi_{6}\rangle+|\psi_{7}\rangle-|\psi_{8}\rangle),\\[1ex]
		|\phi_{3}\rangle =\frac12\,(|\psi_{5}\rangle-|\psi_{6}\rangle-|\psi_{7}\rangle+|\psi_{8}\rangle),\\[1ex]
		|\phi_{4}\rangle =\frac16\,(|\psi_{5}\rangle+|\psi_{6}\rangle+|\psi_{7}\rangle+|\psi_{8}\rangle)-\frac{2\sqrt 2}{3}\,|\psi_{9}\rangle,
	\end{array}
\end{equation}
and five product states
\begin{equation}\label{UPB-3-C-onb}
	\begin{array}{c}
		|\psi_{5}\rangle=\frac1{\sqrt 2}\,|0\rangle|0 + 1\rangle,~~
		|\psi_{6}\rangle=\frac1{\sqrt 2}\,|2\rangle|1 + 2\rangle,\\[1ex]
                                               
		|\psi_{7}\rangle=\frac1{\sqrt 2}\,|0 + 1\rangle|2\rangle,~~
		|\psi_{8}\rangle=\frac1{\sqrt 2}\,|1 + 2\rangle|0\rangle,\\[1ex]

		|\psi_{9}\rangle=|1\rangle|1\rangle\\[1ex]
	\end{array}
\end{equation}
together with $\ket{\psi_1},\, \ket{\psi_2},\, \ket{\psi_3},\, \ket{\psi_4}$ from Eq.~(\ref{UPB-3}) form a complete orthogonal product basis.
Given the ensemble decomposition, we obtain the operator representation of a CP map corresponding to~ the density operator~$\rho$:
 \begin{equation}
	\begin{split}
		M_{\phi_1} = \frac{\sqrt 3}{4\sqrt 2}\begin{pmatrix}1 & -1 & -1\\ 1 & 0 & 1\\-1 & -1 & 1\end{pmatrix}\!,\quad & M_{\phi_2}=\frac{\sqrt 3}{4\sqrt 2}\begin{pmatrix}1 & -1 & -1\\ 1 & 0 & -1\\1 & 1 & -1\end{pmatrix},\\
 		M_{\phi_3} =  \frac{\sqrt 3}{4\sqrt 2}\begin{pmatrix}1 & 1 & 1\\ 1 & 0 & -1\\-1 & -1 & -1\end{pmatrix}\!,\quad & M_{\phi_4} = \frac1{4\sqrt 6}\begin{pmatrix}
		1 & 1 &1\\ 1 & -8 & 1\\ 1 & 1 & 1	
 		\end{pmatrix}.
	\end{split}
\end{equation}
By Eq.~(\ref{OpTransf}), operators $N_{\mu}$ are expressed as follows:
\begin{small}
	\begin{equation}\label{pptn}
		N_{\mu} = \frac1{4\sqrt 6}\begin{pmatrix}3(V_{\mu 1} + V_{\mu 2} + V_{\mu 3}) + V_{\mu 4} \:& \vm 4 + 3(\vm 3 - \vm 2 - \vm 1) \:& \vm 4 + 3(\vm 3 - \vm 1 - \vm 2)\\
			\vm 4 + 3(\vm 1 + \vm 2 + \vm 3) \:& -8\vm 4 \:& \vm 4 + 3(\vm 1 - \vm 2 - \vm 3)\\
			\vm 4 + 3(\vm 2 - \vm 1 - \vm 3) \:& \vm 4 + 3(\vm 2 - \vm 1 - \vm 3) \:& \vm 4 + 3(\vm 1 - \vm 2 - \vm 3)
		\end{pmatrix}.
\end{equation}
\end{small}
The matrix in Eq.~(\ref{pptn}) is of rank 1 when the determinants of all its second order minors vanish. We obtain a system of 9 equations:
\begin{small}
\begin{subequations}
	\begin{align}
	&-3 {{\vm4}^{2}}-10 \vm3\, \vm4-8 \vm2\, \vm4-8 \vm1\, \vm4-3 {{\vm3}^{2}}+3 {{\vm2}^{2}}+6 \vm1\, \vm2+3 {{\vm1}^{2}} \!\!\!\!&= &0,\label{a}\\
	&3 {{\vm4}^{2}}+8 \vm3\, \vm4-10 \vm2\, \vm4-8 \vm1\, \vm4-3 {{\vm3}^{2}}+6 \vm1\, \vm3+3 {{\vm2}^{2}}-3 {{\vm1}^{2}} \!\!\!\!&= &0,\label{bb}\\
	&-\vm3\, \vm4+\vm1\, \vm4-3 {{\vm3}^{2}}-3 \vm2\, \vm3+3 \vm1\, \vm2+3 {{\vm1}^{2}} \!\!\!\!&= &0,\label{c}\\
	&3 {{\vm4}^{2}}-8 \vm3\, \vm4+10 \vm2\, \vm4-8 \vm1\, \vm4-3 {{\vm3}^{2}}-6 \vm1\, \vm3+3 {{\vm2}^{2}}-3 {{\vm1}^{2}} \!\!\!\!&= &0,\label{d}\\
	&-3 {{\vm4}^{2}}+10 \vm3\, \vm4+8 \vm2\, \vm4-8 \vm1\, \vm4-3 {{\vm3}^{2}}+3 {{\vm2}^{2}}-6 \vm1\, \vm2+3 {{\vm1}^{2}} \!\!\!\!&= &0,\label{e}\\
	&\vm3\, \vm4+\vm1\, \vm4-3 {{\vm3}^{2}}-3 \vm2\, \vm3-3 \vm1\, \vm2+3 {{\vm1}^{2}} \!\!\!\!&= &0,\label{f}\\
	&\vm2\, \vm4+\vm1\, \vm4-3 \vm2\, \vm3-3 \vm1\, \vm3+3 {{\vm2}^{2}}-3 {{\vm1}^{2}} \!\!\!\!&= &0,\label{g}\\
	&-\vm2\, \vm4+\vm1\, \vm4-3 \vm2\, \vm3+3 \vm1\, \vm3+3 {{\vm2}^{2}}-3 {{\vm1}^{2}} \!\!\!\!&= &0,\label{h}\\
	&\vm1\, \vm4-3 \vm2\, \vm3 \!\!\!\!&= &0\label{i}.
	\end{align}
\end{subequations}
\end{small}
As it turns out, the system can be easily analyzed.
Adding Eqs.~(\ref{g}), (\ref{h}) and also Eqs.~(\ref{c}), (\ref{f}), with the use of Eq.~(\ref{i}) we have:
\begin{equation}\label{eqv}
	\vm 2^2 = \vm 1^2,\: \vm 3^2 = \vm 1^2.
\end{equation}
Addition of Eqs.~(\ref{bb}), (\ref{d}) yields:
\begin{equation}
	3\vm 4^2 - 8\vm 1\vm 4 - 3\vm 3^2 = 0,
\end{equation}
whereas addition of Eqs.~(\ref{a}), (\ref{e})  gives:
\begin{equation}
	-3\vm 4^2 - 8 \vm 1\vm 4 + 3\vm 1^2 = 0.
\end{equation}
From the last three equations it follows that $\vm 1\vm 4 = 0$, which in combination with Eq.~(\ref{i}) and Eq.~(\ref{eqv}) gives that $\vm 1^2 = 0$, i.~e., the system has only trivial solution. This contradicts with the normalization condition, Eq.~(\ref{ColAbs}). Consequently, by Lemma~\ref{main_lem}, \emph{$\rho$ is entangled}.

\section{\label{sec4}Another point of view: factorization of the density operator}

One can look at the transformations of the operators $\{M_{a}\}$ in Eq.~(\ref{OpTransf}) from another perspective. Let the $m\times n$ matrix $c^{(a)}$ define the $a$-th term of the ensemble decomposition of a given density operator  $\rho$, as in Eq.~(\ref{PureDec}). The $mn\times mn$ matrix of the operator $\rho$ is given by
 \begin{equation}\label{rhodec}
	\rho_{kl,\,mn} = \sum_{a=1}^d\, p_a\, c^{(a)}_{kl}\,(c^{(a)}_{mn})^*,
\end{equation}
where  $d$ -- the number of the terms in the ensemble decomposition of  $\rho$.

Let $\vc A$ denote the \emph{vectorization} of an $m\times n$ matrix  $A$, the  $mn\times 1$ column vector obtained by stacking the rows of the matrix $A$ on top of one another: 
\begin{equation}
\vc A = [a_{11},\ldots,\, a_{1n}, a_{21},\ldots,\,a_{2n},\ldots,\,a_{m 1},\ldots,\,a_{mn}]^T.
\end{equation}
We will refer to the inverse to $\vc A$ operation as \emph{matricization}.

Let $\m$ denote the matrix whose columns are proportional to the vectorizations of $c^{(a)}$: 
\begin{equation}
	\m = \left(\sqrt{p_1}\,\vc {c^{(1)}},\ldots,\, \sqrt{p_d}\,\vc {c^{(d)}}\right).
\end{equation}
The matrix $\m$ has  $mn$ rows and  $d$ columns. Let  $\m_{kl,\,a}$ - the entry of $\m$ on the intersection of the $kl$ row and the $a$-th column. Eq.~(\ref{rhodec}) can be transformed to the matrix factorization of $\rho$:
\begin{equation}
	\rho_{kl,\,mn} = \sum_{a=1}^d\,\m_{kl,\,a}\left(\m_{mn,\,a}\right)^* = \left(\m\,\m^{\dagger}\right)_{kl,\,mn}.
\end{equation}
In other words,
\begin{equation}\label{rhofac}
\rho = \m\,\m^{\dagger},
\end{equation}
and the operator $\rho$ is positive semi-definite, as it should be.

One can multiply $\m$ by an arbitrary unitary matrix:  $\m\to \tilde\m = \m V$ and obtain a new factorization of $\rho$:  $\rho = \tilde\m\,\tilde\m^{\dagger}$, which corresponds to some new ensemble decomposition of the density operator. If the size of $V$ exceeds the number of columns of  $\m$, then $\m$ can be complemented with zero columns up to the appropriate size. The transition $\m\to \tilde\m = \m V$ is none other than the transformation of the operators $\{M_a\}$ in Eq.~(\ref{OpTransf}). One can then matricize the columns of $\tilde\m$ to obtain the new matrices $\tilde c^{(\mu)}$ as well as the new ensemble probabilities $q_{\mu}$. If all $\tilde c^{(\mu)}$ are of unit rank, then $\rho$ is separable, and its separable decomposition is found. Lemma~\ref{main_lem} can be reformulated as follows:

\begin{dotless}
	A density operator $\rho$ is separable if and only if there exists a factorization of its matrix in the form of Eq.~(\ref{rhofac}), where $\m$ - some (in general, rectangular) matrix whose columns are the vectorizations of unit rank matrices.
\end{dotless}

While being hardly more practical than the original lemma, this reformulation can be useful for better understanding of some rank properties of density operators. For example, there is a theorem relating the ranks of a separable density operator and its reductions~\cite{mranks,rusk}:
\begin{unth}
	If a mixed state $\rho\lsc {AB}$ is separable, then 
	\begin{equation}\label{rankd}
		 \rank{\rho\lsc {AB}}\geqslant \max(\rank{\rho\lsc A}, \,\rank{\rho\lsc B}).	
	\end{equation}
\end{unth}
Here we give an alternative proof of this statement.
\begin{proof}
Let $\rho\lsc {AB}$ be separable, then, according to the reformulation of Lemma~\ref{main_lem}, $\rho = \m\,\m^{\dagger}$, where $\m$ -- a matrix whose columns are the vectorizations of unit rank matrices. Let $a^{(i)}$ denote these $m\times n$ matrices. Each of them, being of unit rank, can be factorized as 
\begin{equation}
	a^{(i)} =\sqrt{p_i}\, \xi^{(i)}\eta^{(i)\dagger},
\end{equation}
where $\xi^{(i)},\,\eta^{(i)}$ -- column vectors of unit length and dimensions  $m$ and  $n$, respectively~(general factorization of a rank $1$ matrix has unnormalized vectors, but in this case we have extracted the norms of the vectors explicitly in the factor denoted as $\sqrt{p_i}$ and made them unit length vectors). The reductions of $\rho\lsc {AB}$ on subsystems $A$ and $B$ are then expressed as
\begin{equation}\label{redexp}
	\begin{split}
		\rho\lsc A =\sum_i\, a^{(i)}a^{(i)\dagger} &= \sum_i\,p_i\,\xi^{(i)}\xi^{(i)\dagger},\\
		\rho\lsc B =\sum_i\, \left(a^{(i)\dagger}a^{(i)}\right)^T &= \sum_i\,p_i\,\eta^{(i)*}\left(\eta^{(i)*}\right)^{\dagger},
	\end{split}
\end{equation}
while $\rho\lsc {AB}$ itself is given by
 \begin{equation}\label{denexp}
	 \rho\lsc{AB} = \sum_i\,p_i\,\xi^{(i)}\xi^{(i)\dagger}\,\otimes\,\eta^{(i)*}\left(\eta^{(i)*}\right)^{\dagger},
\end{equation}
where $\eta^{(i)*}$ -- the vector obtained from $\eta^{(i)}$ by complex conjugation of its components.

Consider the selection of the first columns of the matrices  $a^{(i)}$. Each such column is given by the vector 
\[
\sqrt{p_i}(\xi^{(i)}_1\eta^{(i)*}_1,\, \xi^{(i)}_2\eta^{(i)*}_1,\,\ldots,\, \xi^{(i)}_m\eta^{(i)*}_1)^T
.\] 
If $\eta^{(i)}_1 = 0$ for some $i$, then the first column of $a^{(i)}$ is a zero vector. We can apply a unitary transformation $\rho\lsc {AB}\,\to\,(I\lsc A\otimes U\lsc B)\rho\lsc {AB}(I\lsc A\otimes U\lsc B)^{\dagger}$ which, as it follows from Eqs.~(\ref{redexp}), (\ref{denexp}), "rotates" vectors $\eta^{(i)*}$ but doesn't change the global and the local ranks of the density operator $\rho\lsc {AB}$. One can always choose the unitary transformation $U\lsc B$ so that the "rotated"vectors $U\lsc B\eta^{(i)*}$ have nonzero first components. Therefore, with no loss of generality we can assume that the first columns of the matrices $a^{(i)}$ are nonzero vectors. Since $\rank{a^{(i)}} = 1$, the rest $n-1$ columns of each matrix  $a^{(i)}$ are proportional to its first one.

If $\rank{\rho\lsc A} = r\lsc A$, then among the set of the first columns of $a^{(i)}$ there are exactly $r\lsc A$ linearly independent vectors. One can easily see that the first column of $a^{(i)}$ is  a subcolumn of the vectorization of $a^{(i)}$. Therefore, the first columns of $a^{(i)}$ are the subcolumns of the columns of $\m$:
 \begin{equation}
	\m = \begin{pmatrix}
		a\hi 1 _{11} & a\hi 2 _{11} & \ldots\ldots\ldots & a\hi d _{11}\\
		\vdots       & \vdots       & \ldots\ldots\ldots & \vdots\\
		a\hi 1 _{21} & a\hi 2 _{21} & \ldots\ldots\ldots & a\hi d _{21}\\
		\vdots       & \vdots       & \ldots\ldots\ldots & \vdots\\
		\vdots       & \vdots       & \ldots\ldots\ldots & \vdots\\
		a\hi 1 _{m1} & a\hi 2 _{m1} & \ldots\ldots\ldots & a\hi d _{m1}\\
		\vdots       & \vdots       & \ldots\ldots\ldots & \vdots
	\end{pmatrix}
\end{equation}
If $\rank{\rho\lsc A} = r\lsc A$, then $r\lsc A$ subcolumns of  $\m$ are linearly independent, and so are the corresponding columns of  $\m$. Consequently, the number of linearly independent columns of $\m$ is at least  $r\lsc A$, and $\rank{\m} \geqslant r\lsc A$.

Using similar considerations with the set of the first rows of the matrices $a^{(i)}$, one obtains that $\rank{\m}\geqslant\rank{\rho\lsc B}$. The inequality in Eq.~(\ref{rankd}) then follows from the connection between $\m$ and  $\rho\lsc {AB}$ :
\[
\rank{\rho\lsc {AB}} = \rank{\m\,\m^{\dagger}} = \rank{\m}
.\] 
\end{proof}
Note that unlike those of Refs.~\cite{mranks,rusk}, our proof doesn't rely on the reduction and the majorization criteria.
\section{\label{sec5}Conclusions}

As it was shown in section~\ref{sec3}\:\:\:, the method developed in the present paper can be applied to a large variety of states. It seems to be quite operational when the number of terms in the ensemble decomposition of a density operator is small, but still in this case it can't be better than the Peres-Horodecki criterion which is known to be necessary and sufficient for low rank density operators~\cite{Horod2}. In the low rank case it can serve as a complementary method which allows to obtain separable decompositions of density matrices. In the high rank case, as examples with isotropic and PPT states showed, our approach is by no means operational, and it demands a detailed analysis of systems of polynomial equations. In some cases these polynomial systems can be reduced to simple ones; besides, computer algebra methods can be applied. The presented approach is also interesting from the theoretical point of view and can be used in derivation of various separability criteria.

The open question is the connection with the Peres-Horodecki criterion  in the qubit-qubit case. If each term in the ensemble decomposition of a given density operator is defined, as in Eq.~(\ref{PureDec}), by $c^{(a)}_{ij}$, a $2\times 2$ matrix in this case, then the transformed operators $\{N_{\mu}\}$ have the following matrix representation:
\begin{equation}
	(N_{\mu})_{ij} = \sqrt m\,\sum_a \, \vm a\, \sqrt{p_a}\, c^{(a)}_{ji},
\end{equation}
and the unit rank condition yields only one equation for each $\mu$ in the qubit-qubit case:
\begin{equation}
\sum_{a,\, b}\,\sqrt{p_a \,p_b}\,V_{\mu a} V_{\mu b}\, (c^{(a)}_{00}c^{(b)}_{11} - c^{(a)}_{01}c^{(b)}_{10}) = 0.	
\end{equation}
For a $2\otimes 2$ state positivity of a partial transpose is sufficient for separability, but it is unclear how this fact implies existence of the matrix $V_{\mu a}$ having orthonormal columns and satisfying the above equation.

\bigskip

\noindent\emph{Acknowledgements} - The author would like to thank G.~G.~Amosov, S.~N.~Filippov, M.~E.~Shirokov, and A.~M.~Chebotarev for inspiring discussions and comments on an earlier version of the manuscript. The author thanks Lomonosov Moscow State University for supporting this work.

\end{document}